\documentclass[pra,showpacs,graphics,twocolumn,floatfix,mathbbm,a4paper,nofootinbib]{revtex4-1}

\usepackage{amsthm}
\usepackage{amsmath}
\usepackage{latexsym}
\usepackage{amsfonts}
\usepackage{amssymb}
\usepackage{color}
\usepackage{bbm,dsfont}
\usepackage{graphicx}
\usepackage{subfigure}
\usepackage{mathrsfs}
\usepackage{mathbbol}
\usepackage{hyperref}
\usepackage{enumerate}
\usepackage{MnSymbol}

\newtheorem{proposition}{Proposition}
\newtheorem{proposition?}{Proposition?}

\newtheorem{corollary}{Corollary}

\theoremstyle{definition}

\newtheorem{example}{Example}

\newcommand{\real}{\mathbb R} %real
\newcommand{\integer}{\mathbb Z} %integer
\newcommand{\hi}{\mathcal{H}} %Hilbert space H
\newcommand{\lh}{\mathcal{L(H)}} %bounded linear operators
\newcommand{\sh}{\mathcal{S(H)}} %states on H
\newcommand{\ip}[2]{\left\langle\,#1\,|\,#2\,\right\rangle} %inner product
\newcommand{\ket}[1]{|#1\rangle} %ket
\newcommand{\bra}[1]{\langle#1|} %bra
\newcommand{\kb}[2]{|#1\rangle\langle#2|} %ketbra
\newcommand{\id}{\mathbbm{1}} %identity operator
\newcommand{\A}{\mathsf{A}}%generic observable
\newcommand{\B}{\mathsf{B}}%generic observable
\newcommand{\C}{\mathsf{C}}%generic observable
\newcommand{\D}{\mathsf{D}}%generic observable
\newcommand{\G}{\mathsf{G}}%generic joint observable
\newcommand{\M}{\mathsf{M}}%joint obs
\newcommand{\T}{\mathsf{T}}
\newcommand{\N}{\mathsf{N}}

\newcommand{\state}{\mathcal{S}} %states
\newcommand{\effect}{\mathcal{E}} %effects
\newcommand{\obs}{\mathcal{O}} %observables
\newcommand{\noise}{\mathcal{N}} %noise
\newcommand{\trivial}{\mathcal{T}} %trivial

\newcommand{\tgeq}{\stackrel{t}{\geq}}

\begin{document}

\title[]{A necessary condition for incompatibility of observables in general probabilistic theories}

\author{Sergey N. Filippov}
\email{sergey.filippov@phystech.edu}
\address{Institute of Physics and Technology, Russian Academy of Sciences, Moscow, Russia}
\address{Moscow Institute of Physics and Technology, Dolgoprudny, Moscow Region, Russia}

\author{Teiko Heinosaari}
\email{teiko.heinosaari@utu.fi}
\address{Turku Centre for Quantum Physics, Department of Physics and Astronomy, University of Turku, 20014, Finland}

\author{Leevi Lepp\"{a}j\"{a}rvi}
\email{leille@utu.fi}
\address{Turku Centre for Quantum Physics, Department of Physics and Astronomy, University of Turku, 20014, Finland}

\pacs{03.65.Ta}

\begin{abstract}
We quantify the intrinsic noise content of an observable in a general probabilistic
theory and derive a noise content inequality for incompatible
observables. We apply the derived inequality to standard quantum theory, the quantum theory of processes, and polytope state
spaces. The noise content for positive operator-valued measures
takes a particularly simple form and equals the sum of minimal
eigenvalues of all the effects. We illustrate our findings with a
number of examples including the introduced notion of reverse observables.
\end{abstract}

\maketitle

%%%%%%%%%%%%%%%%%%%%%%%%%%%%%%%%
%%%%%%%%%%%%%%%%%%%%%%%%%%%%%%%%
%%%%%%%%%%%%%%%%%%%%%%%%%%%%%%%%
%%%%%%%%%%%%%%%%%%%%%%%%%%%%%%%%
%%%%%%%%%%%%%%%%%%%%%%%%%%%%%%%%

%%%%%%%%%%%%%
\section{Introduction}
%%%%%%%%%%%%%

Quantum theory can be considered as a particular instance within a
wide range of probabilistic theories~\cite{barret-2007,chiribella-2010}.
On the one hand, quantum
theory inherits general properties of probabilistic theories and,
consequently, one may deduce some features already from a general
operational framework.
For instance, the limitations on broadcastable subsets of states can be derived from this generality~\cite{barnum-2007}.
On the other hand, particular properties of quantum theory, like specific constraints on nonlocality, partially fix its position with respect to other probabilistic theories~\cite{allcock-2009}.
As a result, specification of information-theoretic axioms may be sufficient for quantum theory to be derived~\cite{foils}.

A general probabilistic theory operates with notions of states and
observables. The set of states $\state$ is convex since any
probabilistic mixture of states must be a valid state. Observables
are then affine functionals from the set of states $\state$ to the
set of probability distributions. In the standard quantum theory,
states are associated with density operators, whereas observables
are mathematically described by positive operator-valued measures
(POVMs)~\cite{bush-book,holevo-book,heinosaari-ziman-book}.
However, when we are testing a quantum process, then quantum
channels are the examined objects and they are hence regarded as
states, whereas observables can be described by process
POVMs~\cite{chiribella-2008,ziman-2008,chiribella-2009}.
Theories describing the Popescu--Rohrlich (PR) box~\cite{popescu-1994}
and polytope state spaces serve as other examples of general probabilistic
theories~\cite{kimura-2009,janotta-2011}.

A set of observables in a general probabilistic theory may possess the
property of being incompatible, which means that those observables
cannot be seen as components of a single observable
\cite{stevens-2014,banik-2015,heinosaari-miyadera-ziman-2016,sedlak-2016}.
Incompatibility is a non-classical feature, since in a general
probabilistic theory with a classical state space all observables
are compatible, while every non-classical theory possesses some
incompatible observables \cite{aravinda-2016,plavala-2016}. It is possible to
compare the incompatibility of finite sets of observables in
different probabilistic physical theories in a quantitative way
\cite{busch-2013,gudder-2013,banik-2013}. Interestingly,
quantum theory contains maximally incompatible pairs of
observables, but only when the underlying Hilbert space is
infinite dimensional \cite{heinosaari-schultz-2014}.

This work focuses on incompatibility of observables in general
probabilistic theories. The main goal of the present investigation
is to quantify the noise content for observables in general
probabilistic theories and to exploit it in deriving a sufficient
condition for compatibility, i.e., a necessary condition for
incompatibility for a collection of observables. To demonstrate
that the derived condition is noteworthy, we use it to formulate a
readily verifiable necessary condition for incompatibility in
quantum theory. To anticipate this result, the condition takes the
following form for POVMs: \textit{If $m$ POVMs are incompatible,
then the sum of minimal eigenvalues of all their elements is less
than $m-1$.} We illustrate our findings by a number of examples
including a newly introduced class of reverse observables.
Consideration of standard quantum theory is followed by
theories with quantum processes as states, as well as the square
bit state space.

We note that in the case of POVMs, noise robustness of incompatibility has been investigated in several recent works \cite{haapasalo-2015,zhu-2015,kiukas-2015,uola-2016}.
The conditions found in those works are tighter than the condition presented in this work, but this is due to the fact that they are applicable only for POVMs with some specific structure or symmetry.
Moreover, in contrast to most of the earlier studies (see for example \cite{stevens-2014, heinosaari-miyadera-ziman-2016, gudder-2013, wolf-2009}), we do not add noise to given
observables but rather look for the intrinsic noise which is already present.
We show that a meaningful nontrivial noise inequality can be derived already at the level of a general probabilistic theory.

The paper is organized as follows. In
Sec.~\ref{sec:incompatibility} the incompatibility of observables
in general probabilistic theories is reviewed. In Sec. \ref{sec:condition} the noise content in observables is defined, and a sufficient condition
for compatibility of a set of observables is formulated.
The usage of the general condition is then demonstrated in Sec. \ref{sec:specific}.

%%%%%%%%%%%%%
\section{Incompatibility of observables in general probabilistic
theories} \label{sec:incompatibility}
%%%%%%%%%%%%%

%%%%%%%%%%%%%
\subsection{States, effects, and observables}
%%%%%%%%%%%%%

We begin by recalling the basic elements of the standard framework of general probabilistic theories (see e.g. \cite{beltrametti-1997,barnum-2011} for more detailed presentations).
In a general probabilistic theory, the set of states $\state$ is a convex subset of a finite dimensional real vector space $V$.
The convexity is a result of the probabilistic nature of the theory, meaning that the convex sum $p s_1 + (1-p) s_2$ is a state whenever $s_1,s_2$ are states and $0 \leq p \leq 1$.

We denote by $F(\state)$ the linear space of all affine functionals from $\state$ to $\mathbb{R}$, i.e., a functional $e:\state \to \mathbb{R}$ is in $F(\state)$ if it satisfies
$$
e( p s_1 + (1-p) s_2) = p e(s_1) + (1-p) e(s_2) 
$$
for all $s_1,s_2\in\state, 0 \leq p \leq 1$.
For two functionals $e,f \in F(\state)$, we denote $e \leq f$ if $e(s) \leq f(s)$ for all $s \in \state$. We further denote by $u \in F(\state)$ the unit map satisfying $u(s) = 1$ for all $s \in \state$.
The set of effects on $\state$ is defined as
\begin{align*}
\effect(\state) = \{ e \in F(\state) : 0 \leq e \leq u \} \, ,
\end{align*}
i.e., it is the convex subset of those affine functionals $e$ for which $0 \leq e(s) \leq 1$ for all $s \in \state$.
The set of effects arising as functionals on states is a particular example of an effect algebra \cite{beltrametti-1997}.
In particular, $\effect(\state)$ has a partially defined sum $e+f$, which is simply the functional addition of $e$ and $f$ defined whenever $e+f \leq u$.

An observable with a finite number of outcomes is a
function $\A: x \mapsto \A_x$ from a finite outcome set $X\subset\integer$ to
$\effect(\state)$.
The number $\A_x(s)$ is interpreted as the probability
of getting the outcome $x$ in a measurement of the observable $\A$ when the system is in
the state $s$.
As we must have $\sum_{x \in X}
\A_x(s) = 1$ for all $s \in \state$, we have the normalization
condition $\sum_{x \in X} \A_x = u$.
We denote the set of observables with an outcome set $X$ by $\obs_X$, and by $\obs$ the set of all observables with a finite number of outcomes.

A special type of observable is the \emph{trivial observable} $\T$, which is such that for each outcome $x$, $\T_x(s) = \T_x(s')$ for all $s,s'\in\state$.
We denote the set of trivial observables by $\mathcal{T}$.
Since the outcome probabilities for a trivial observable are the same for all states, it does not provide any information on an input state. 

In what follows we recall the two most important instances of general probabilistic theories, standard quantum theory and the quantum theory of processes.

\begin{example}[\emph{Quantum theory}]
Let $\state_q$ be the convex set of density operators $\varrho$ on a Hilbert space $\hi$. Then the set of effects $\effect(\state_q)$, defined as affine mappings on $\state_q$, can be represented as $e(\varrho) = {\rm tr}[\varrho E]$ for all states $\varrho$, where $E$ is a selfadjoint operator satisfying the operator inequalities $0\leq E \leq \id$. This correspondence is one-to-one, so effects can be identified with these effect operators.
With this identification, an observable $\A: x \mapsto \A_x$ with a finite outcome set $X$ is a POVM satisfying $\sum_{x \in X} \A_x = \id$.
A trivial observable $\T$ is of the form $\T_x = p_x \id$, where $p_x$ is a probability distribution on $X$.
\end{example}

\begin{example}[\emph{Quantum theory of processes}]
We denote by $\lh$ the bounded linear operators on a Hilbert space $\hi$.
Let $\state_p$ be the set of completely positive and trace
preserving maps $\Phi:\mathcal{L}(\hi_A) \mapsto \mathcal{L}(\hi_B)$, called quantum channels or processes. Then the set of
effects $\effect(\state_p)$ can be represented as the set of
operators $M$ on $\hi_A \otimes \hi_B$ satisfying $0 \leq M \leq
\varrho \otimes \id$ for some density operator $\varrho$ on
$\hi_A$. This representation is given as $e(\Phi) = {\rm
tr}[\Omega_{\Phi} M]$, where $\Omega_{\Phi}$ is the Choi operator
of $\Phi$, i.e., $\Omega_{\Phi}=(id \otimes \Phi)[\ket{\psi_+}
\bra{\psi_+}]$, where $\psi_+= \sum_{i=1}^{d} \phi_i \otimes
\phi_i$ and $\{\phi_i\}_{i=1}^d$ is an orthonormal basis of
$\hi_A$. An important point is that this correspondence between
affine maps and operators is not one-to-one; two operators $M$ and
$M'$ correspond to the same effect $e$ exactly when $M-M' = \omega
\otimes \id$ for some traceless operator $\omega$
\cite{jencova-2012,jencova-2013}. In this representation an
observable $\A: x \mapsto \A_x$ with a finite outcome set $X$
satisfies the normalization $\sum_{x \in X} \A_x = \varrho \otimes
\id$ for some density operator $\varrho$ on $\hi_A$. This kind of
map is called a process POVM, or PPOVM for short
\cite{ziman-2008}. A trivial PPOVM is of the form $\T_x = p_x
\xi_x \otimes \id$, where each $\xi_x$ is a density operator on
$\hi_A$ and $p_x$ is a probability distribution. Two trivial
PPOVMs $\T_x = p_x \xi_x \otimes \id$ and $\T'_x = p'_x \xi'_x
\otimes \id$ correspond to the same trivial observable exactly
when the probability distributions $p_x$ and $p'_x$ are the same.
\end{example}

%%%%%%%%%%%%%
\subsection{Post-processing of observables}
%%%%%%%%%%%%%

A classical channel $\nu$ between outcome spaces $X$ and $Y$ is a
right stochastic matrix with elements $\nu_{xy}$, $x \in X$, $y
\in Y$, i.e., $0\leq \nu_{xy} \leq 1$ and $\sum_{y \in Y} \nu_{xy}
= 1$. The number $\nu_{xy}$ is the transition probability for an
element $x$ to be transformed into $y$. Classical channels are
often used to describe noise, but we can also think of a classical
channel as an active transformation that is implemented on
outcomes. In the following we recall two classes of classical
channels that will be used later.

%%%%%%%%%%%%%%%%%%%%%%%%%%%%%%%%%%%%%%%%%%%%%%%%%%%%%%%%%%%%%%%%%%%%
\begin{figure}
\includegraphics[width=8cm]{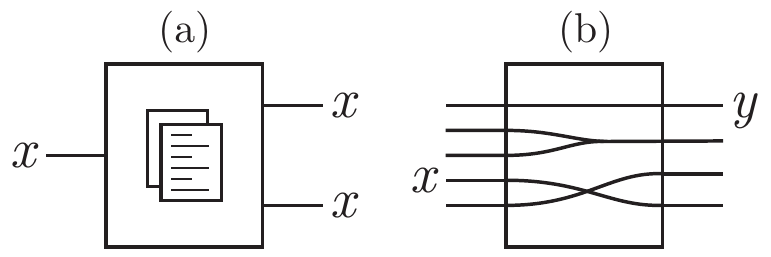}
\caption{\label{figure1} (a) Action of the classical copying
channel. (b) Example of the relabeling channel.}
\end{figure}
%%%%%%%%%%%%%%%%%%%%%%%%%%%%%%%%%%%%%%%%%%%%%%%%%%%%%%%%%%%%%%%%%%%%

\begin{example}[\emph{Copying the measurement outcomes}]
Measurement outcomes are just classical symbols and thus can be
copied. To see copying as a classical channel, let $Y=X \times X$.
The stochastic matrix $\nu^{\rm c}_{xy}$ related to copying is
defined as $\nu^{\rm c}_{xy}=1$ if $y=(x,x)$ and $\nu^{\rm
c}_{xy}=0$ otherwise. This transforms any $x$ to $(x,x)$.
Fig.~\ref{figure1}(a) depicts the action of the copying channel.
Multiple applications of a copying channel allows one to make an
arbitrary number of copies of an outcome $x$. If the number of
copies equals $m$, then we call it an $m$-copying channel.
\end{example}

\begin{example}[\emph{Relabeling the measurement outcomes}]
The copying channels belong to a wider class of classical channels
where measurement outcomes are relabeled deterministically into
some other outcome. Let $f:X\to Y$ be a relabeling function. The
derived stochastic matrix $\nu^f_{xy}$ is defined as
$\nu^f_{xy}=1$ if $f(x)=y$ and $\nu^f_{xy}=0$ otherwise. In
contrast to the copying procedure, generally several outcomes can
be relabeled into a single new outcome, see Fig.~\ref{figure1}(b).
\end{example}

Let $\A$ be an observable with an outcome set $X$ and let $\nu$ be a classical channel between $X$ and some other outcome space $Y$.
We denote by
$\nu\circ\A$ the new observable defined as
\begin{align}
(\nu \circ \A)_y = \sum_{x \in X} \nu_{xy} \A_x
\end{align}
for all outcomes $y \in Y$.
Physically, the observable $\nu\circ\A$ is implemented by first measuring $\A$ and then using the classical channel $\nu$ on each obtained measurement outcome.
This way of forming new observables gives rise to a preorder in the set of observables \cite{martens-1990,buscemi-2005,heinonen-2005}.
Namely, for two observables $\A$ and $\B$, we say that \emph{$\B$ is
a post-processing of $\A$} if there exists a classical channel
$\nu$ such that $\B = \nu \circ \A$.

%%%%%%%%%%%%%%%%%%%%%%%%%%%%%%%%%%%%%%%%%%%%%%%%%%%%%%%%%%%%%%%%%%%%
\begin{figure}
\includegraphics[width=8cm]{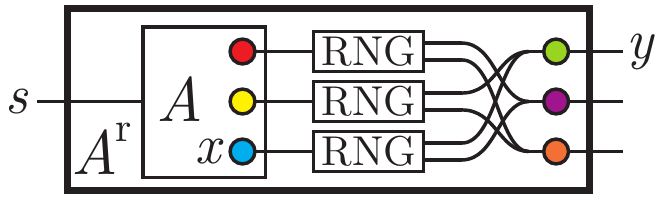}
\caption{\label{figure2} $\A^{\rm r}$ is the reverse observable
with respect to $\A$. Outcome $y$ of observable $\A$ does not
contribute to the outcome $y$ of observable $\A^{\rm r}$, so they
are illustrated by complementary colors. RNG stands for a random
number generator which uniformly chooses outcome $y \neq x$.}
\end{figure}
%%%%%%%%%%%%%%%%%%%%%%%%%%%%%%%%%%%%%%%%%%%%%%%%%%%%%%%%%%%%%%%%%%%%

\begin{example}
(\emph{Reverse observable}.) A reversing channel is a classical
channel $\nu^{\rm r}: X \mapsto X$ such that $\nu^{\rm r}_{xy} =
0$ if $x=y$ and $\nu^{\rm r}_{xy} = \nu^{\rm r}_{x'y}$ for all
$x,x' \neq y$. If the outcome set $X$ contains $N$ elements, then
$\nu^{\rm r}_{xy} = \nu^{\rm r}_{x'y} = \frac{1}{N-1}$ for all
$x,x' \neq y$. For each observable $\A$, the observable $\A^{\rm
r} = \nu^{\rm r} \circ \A$ is called the reverse version of $\A$.
If $\A$ has $N$ outcomes, then the reverse observable $\A^{\rm r}$
takes the form
\begin{equation}\label{eq:not-2}
\A^{\rm r}_x = \frac{1}{N-1} \sum_{y \neq x} \A_y = \frac{1}{N-1}( u - \A_x)\, .
\end{equation}
The physical meaning of $\A^{\rm r}$ is illustrated in
Fig.~\ref{figure2}. After $\A$ has been measured and outcome $x$
has been obtained, we roll a fair dice with $N-1$ sides and randomly choose any outcome $y$ different from $x$. This is taken to be
the outcome of the new observable $\A^{\rm r}$, which is hence
given by formula \eqref{eq:not-2}.
\end{example}

\begin{example}\label{ex:doubly}
(\emph{Doubly reverse observable}.) Performing the reversing
postprocessing two times, we get
\begin{equation}
\label{eq:doubly-reverse} \A_y^{\rm rr} = \frac{1}{(N-1)^2} \left[
\A_y + (N-2) u\right],
\end{equation}

\noindent or, concisely, $\A^{\rm rr} = (1-\lambda) \A + \lambda
\T$, where $\lambda = \frac{N(N-2)}{(N-1)^2}$ and $\T$ is the
trivial observable with uniform distribution of outcomes. In the
case of two outcomes ($N=2$), the doubly reverse observable
coincides with the original one, i.e., $A^{\rm rr} = A$.
\end{example}

As one would expect, a trivial observable $\T$ is a
post-processing of any other observable $\A$. To see this, we
define a classical channel $\nu^\T$ as $\nu^\T_{xy}=\T_y(s_0)$ for
all $x$, where $s_0$ is any state. Then
\begin{align*}
(\nu^\T \circ \A)_y(s) & = \sum_{x \in X} \nu^\T_{xy} \A_x(s) = \sum_{x \in X} \T_y(s_0) \A_x(s) = \T_y(s_0) \\
& = \T_y(s) \, ,
\end{align*}
showing that $\nu^\T \circ \A = \T$.
The classical channel $\nu^\T$ just erases the outcome obtained in the $\A$-measurement, and replaces it with a new outcome according to the measurement outcome distribution of $\T$, which is the same for all states.

%%%%%%%%%%%%%
\subsection{Incompatibility of observables}
%%%%%%%%%%%%%

A collection of observables $\mathcal{P}$ is \emph{compatible} if
there exists an observable $\C$, with an outcome set $Y$, such
that each observable $\A\in\mathcal{P}$ is a post-processing of
$\C$. A compatible collection of observables can thus be
implemented simultaneously by first measuring $\C$, then copying
the classical outcomes, and finally applying the relevant
post-processings to the copied outcomes. This definition is
depicted in Fig.~\ref{figure3}. If a set of observables is not
compatible, then it is called \emph{incompatible}.

%%%%%%%%%%%%%%%%%%%%%%%%%%%%%%%%%%%%%%%%%%%%%%%%%%%%%%%%%%%%%%%%%%%%
\begin{figure}
\includegraphics[width=8cm]{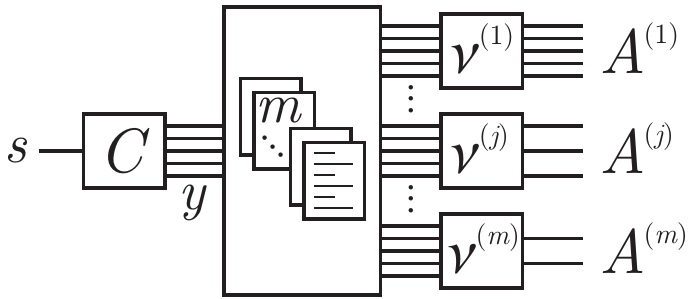}
\caption{\label{figure3} Observables $\A^{(1)}, \ldots, \A^{(m)}$
are compatible if each of them is a post-processing of some
observable $\C$.}
\end{figure}
%%%%%%%%%%%%%%%%%%%%%%%%%%%%%%%%%%%%%%%%%%%%%%%%%%%%%%%%%%%%%%%%%%%%

Let $\{ \A^{(1)}, \ldots, \A^{(m)} \}$ be a compatible set of $m$
observables, with outcome sets $X^{(1)},\ldots,X^{(m)}$, respectively.
Thus, there exists an observable $\C$ and classical channels $\nu^{(1)}, \ldots, \nu^{(m)}$ such that
\begin{equation}\label{eq:mother-obs}
\A^{(j)} = \nu^{(j)} \circ \C, \quad j = 1, \ldots, m \, .
\end{equation}
To see this definition of compatibility in an equivalent form, we denote
\begin{equation}
\G_{ x^{(1)} \ldots x^{(m)} } = \sum_{y} \prod_{j=1}^{m}
\nu_{yx^{(j)}}^{(j)} \C_{y}
\end{equation}
for all $x^{(j)} \in X^{(j)}$, $j = 1, \ldots, m$. Then $\G$ is an
observable, and from \eqref{eq:mother-obs} it follows that
\begin{equation}
\label{eq:joint}
\A_{x^{(j)}}^{(j)} = \sum\limits_{x^{(i)}: \, i
\neq j} \G_{ x^{(1)} \ldots x^{(m)} }.
\end{equation}
Thus, the compatibility of observables $\A^{(1)}, \ldots, \A^{(m)}$
implies that there exists a \emph{joint
observable} $\G$ with the outcome space $X^{(1)} \times \cdots \times
X^{(m)}$ such that the observables are marginals of the joint observable.
Conversely, starting from $\G$ and taking classical channels corresponding to relabeling functions that are projections, $pr_\ell:X^n\to X$, $pr_\ell(x_1,\ldots,x_n) = x_\ell$, we see that \eqref{eq:joint} is a special case of \eqref{eq:mother-obs}.
As noted in \cite{ali-2009} in the case of quantum observables, we conclude that \emph{a subset of observables is compatible if and only if they have a joint observable}.
The latter condition is usually taken as the definition of joint measurability of quantum observables \cite{lahti-2003}.

%%%%%%%%%%%%%
\section{Necessary condition for incompatibility}\label{sec:condition}
%%%%%%%%%%%%%

%%%%%%%%%%%%%
\subsection{Noise content of an observable}
%%%%%%%%%%%%%

In order to formulate a necessary condition for incompatibility of
observables, we first quantify their intrinsic fuzziness, or
noise content, and then use the extraction of that noise in an
explicit construction of a class of joint observables.

In a general probabilistic theory, one can introduce a procedure of
mixing observables. Suppose $\A: X \to \effect(\state)$ and $\B: Y
\to \effect(\state)$ are observables with outcome sets $X$ and $Y$, respectively.
Then a mixture of $\A$ and
$\B$, with a mixing parameter $0
\leq t\leq 1$, is an observable $\C: X \cup Y \to \effect(\state)$ such that
\begin{equation}
\label{mixture} \C_z = t\A_z + (1-t) \B_z\end{equation}
for all $z \in X \cup Y$, where $\A$ and $\B$ can be extended to $X\cup Y$ by defining $\A_z=0$ if $z \nin X$ and $\B_z=0$ if $z \nin Y$.

We are interested in a situation where one of the observables in the
right-hand side of mixture \eqref{mixture} is not arbitrary but
belongs to a some specified subset $\noise \subseteq \obs$ which describes noise in the measurement. If the
target observable $\C$ is not in $\noise$, then this requirement
imposes limitations on possible values of the mixing parameter
$t$.

For the following consideration, we fix a nonempty subset $\noise \subseteq \obs$ which describes noisy observables.
Then, the physical
meaning of Eq.~\eqref{mixture} is to decompose an observable into
its noisy part and the rest.
A quantitative description of the noise content is attained by
maximizing $t$.
Therefore, for each observable $\A$, we denote
\begin{align}\label{eq:noise}
\nonumber w(\A;\noise) &= \sup\{ 0\leq t\leq 1 :  t\N + (1-t) \B=\A \\
 & \quad \quad \quad \quad \quad  \textrm{for some }\N\in\noise \textrm{ and }\B\in\obs \}
\end{align}
and call this quantity the \emph{noise content of $\A$ with respect to $\noise$}.
We note that the observables $\N$ and $\B$ in \eqref{eq:noise} can be assumed to have the same outcome set as $\A$.

Whenever $\A_x \geq t \N_x$ for some $0\leq t \leq 1$ for all $x \in X$, we will use the notation $\A \tgeq
\N$. Suppose $0 \leq t < 1$ and $\A \tgeq \N$, then we can write
$\A$ as a mixture
\begin{equation}\label{eq:A=tB}
\A = t \N + (1-t) \widetilde{\A} \, ,
\end{equation}

\noindent where $\widetilde{\A}$ is the observable defined as
\begin{equation}
\label{less-fuzzy} \widetilde{\A} = (1-t)^{-1}(\A - t \N) \, .
\end{equation}

\noindent Conversely, if there exists some observable
$\widetilde{\A}$ such that \eqref{eq:A=tB} holds, then $\A \tgeq
\N$. Thus, one can reformulate the definition of noise content of
$\A$ with respect to $\noise$ as follows:
\begin{align}
\label{eq:noise2} w(\A;\noise) = \sup\{ 0\leq t \leq 1 :  \A \tgeq
\N \textrm{ for some }\N\in\noise\}.
\end{align}

Specific properties of the map $\A \mapsto w(\A;\noise)$ depend on
the choice of the subset $\noise$. There are, however, some
general features valid for any noise set $\noise$. In
particular, we observe the following:
\begin{itemize}
\item[(a)] If $\nu$ is a classical channel and $\nu \circ \noise
\subseteq \noise$, then $w(\nu \circ \A;\noise) \geq
w(\A;\noise)$.

\item[(b)] If $\noise$ is convex, then $w(s \A + (1-s)
\B;\noise) \geq s w(\A;\noise) + (1-s) w(\B;\noise)$ for all
$0\leq s \leq 1$.
\end{itemize}
The first property follows directly from the definition of $w(\A;\noise)$, while the latter is seen to be valid by first noticing that
$$
s \A_x + (1-s) \B_x \geq s w(\A;
\noise) \N_x + (1 - s) w(\B; \noise) \M_x
$$
for some observables $\N,\M \in \noise$ and all outcomes $x$.
We denote
\begin{align*}
p_{\A} &= s w(\A; \noise) / [s w(\A; \noise) + (1 - s) w(\B; \noise)] \, , \\
p_{\B} &= (1- s) w(\B; \noise) / [s w(\A; \noise) + (1 - s) w(\B; \noise)]
\end{align*}
and then obtain
$$
s \A_x + (1-s) \B_x \geq [ s w(\A; \noise) + (1 - s)
w(\B; \noise) ] (p_{\A} \N +
p_{\B} \M)_x \, ,
$$
where $p_{\A} \N +p_{\B} \M \in \noise$ as the set $\noise$ is convex.

The prototypical choice for $\noise$ is to take $\noise = \trivial$, the set of all trivial observables.
In this case, we simply say that $w(\A;\trivial)$ is the noise content of $\A$.
The set $\trivial$ is convex and $\nu \circ \trivial \subseteq \trivial$ for all classical channels.

\begin{proposition}\label{prop:inf}
Let $\A$ be an observable on a finite outcome set $X$. Then $w(\A;\trivial) = \sum_{x \in X}  \inf_{s \in \state} \A_x(s)$.
\end{proposition}

\begin{proof}
Denote  $a_x = \inf_{s \in \state} \A_x(s)$ and $a=\sum_x a_x$. First assume that $a_x =0$ for all $x\in X$ so that $a=0$. Let $\T\in\trivial$ be a trivial observable and take any $t$, $0 \leq t \leq 1$, such that $\A \tgeq \T$. By our definitions this is equivalent to $\A_x(s) \geq t \T_x(s)$ for all $x\in X$ and $s \in \state$, so that for all $x \in X$ we have that
\begin{align*}
0 = a_x = \inf_{s \in \state} \A_x(s) \geq t \inf_{s \in \state} \T_x(s) = t p_x,
\end{align*}
where $p_x \equiv \T_x(s)$ is the probability distribution defined by $\T$. Summing over $x$ we get
\begin{align*}
0 = a = \sum_x a_x \geq t \sum_x p_x = t.
\end{align*}
Since also $0\leq t \leq 1$, we must have $t=0$, and since this holds for all $\T \in \trivial$, by \eqref{eq:noise2} we get that $w(\A;\trivial)=a=0$.

Secondly, assume that $a_x \neq 0$ at least for some $x \in X$. By similar arguments as above, we see that for all $x \in X$ we have $a_x \geq t' p'_x$, where $p'_x = \T'_x(s)$ is a probability distribution defined by some trivial observable $\T'\in \trivial$ for some $0\leq t'\leq 1$. Summing over all $x$ we then get an upper bound for $t'$ as $a = \sum_x a_x \geq t'$. We see that the upper bound is attained if we define $\T'$ as $\T'(s)= p'_x=a_x/a$. Thus by \eqref{eq:noise2} we have that $w(\A; \trivial) =a$.
\end{proof}

%%%%%%%%%%%%%
\subsection{Joint measurement scheme}
%%%%%%%%%%%%%

The joint measurement scheme that we will next discuss is an elaboration of the one presented in \cite{uola-2016}.
The idea is that we first write the definition of compatibility in a slightly different way, then limit the defining conditions, and in this way we obtain a computable sufficient condition for compatibility.

From the definition, two observables $\A$ and $\B$ are compatible if
there exists a third observable $\C$ and classical channels
$\nu_1$ and $\nu_2$ such that $\A = \nu_1 \circ \C$ and $\B =
\nu_2 \circ \C$. Let us consider a seemingly more general scheme,
where we are asking for the existence of two observables $\C$ and
$\D$, classical channels $\nu_1$, $\nu_2$, $\mu_1$ and $\mu_2$,
and a mixing parameter $t$ such that
\begin{align}
\A & = t\nu_1 \circ \C + (1-t) \mu_1 \circ \D \label{eq:A-gen}\\
\B & = t\nu_2 \circ \C + (1-t) \mu_2 \circ \D \, . \label{eq:B-gen}
\end{align}
Thus, $\A$ and $\B$ are now required to be mixtures of post-processings of $\C$ and $\D$; see Fig. \ref{figure4}.

Clearly, the conditions  \eqref{eq:A-gen}--\eqref{eq:B-gen} reduce
to the usual compatibility conditions when $t=1$. Therefore, every
compatible pair can be written in this new form. Conversely, if
two observables $\A$ and $\B$ can be written in the form
\eqref{eq:A-gen}--\eqref{eq:B-gen}, then they are compatible.
In fact, $\A$ and $\B$ are post-processings of the mixed
observable $t \C + (1-t) \D$, but now the mixture has an extra outcome to keep track of which observable was measured each time.
After measuring either $\C$ or $\D$, we duplicate the outcome and post-process with either $\nu_1$ and $\nu_2$ or $\mu_1$ and $\mu_2$, depending on the measured observable.

%%%%%%%%%%%%%%%%%%%%%%%%%%%%%%%%%%%%%%%%%%%%%%%%%%%%%%%%%%%%%%%%%%%%
\begin{figure}
\includegraphics[width=8.5cm]{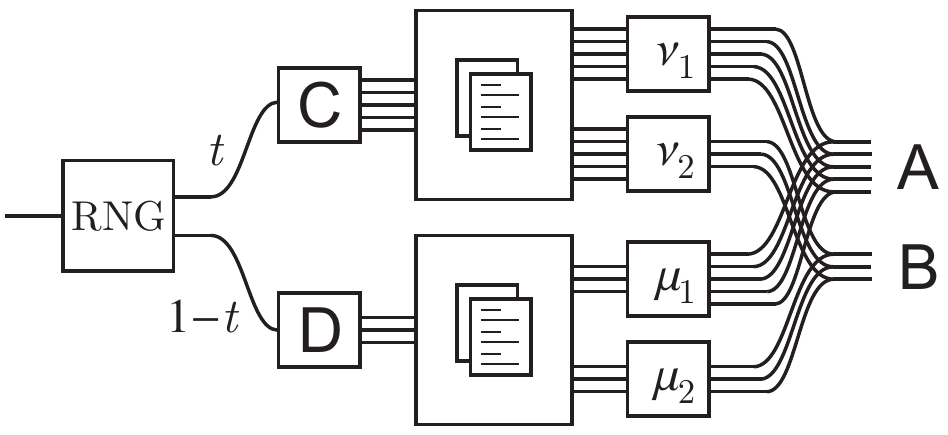}
\caption{\label{figure4} The considered joint measurement
scheme for two observables $\A$ and $\B$ consists of a random choice between two observables $\C$ and $\D$, followed by separated post-processing for both $\A$ and $\B$ that aim to approximate these observables.}
\end{figure}
%%%%%%%%%%%%%%%%%%%%%%%%%%%%%%%%%%%%%%%%%%%%%%%%%%%%%%%%%%%%%%%%%%%%

%%%%%%%%%%%%%
\subsection{Incompatibility inequality}
%%%%%%%%%%%%%

As a special case of the joint measurement scheme described previously, we limit the choice of classical channels $\mu_1$  and $\nu_2$ to those that make observables  $\mu_1 \circ \D$ and $\nu_2 \circ \C$ trivial.
Since any trivial observable is a post-processing of any other observable, we get all trivial observables, irrespective of $\C$ and $\D$.
Hence, the conditions  \eqref{eq:A-gen}--\eqref{eq:B-gen} reduce to
\begin{align}
\A & = t \nu_1 \circ \C + (1-t) \T_1 \label{eq:A-spes} \\
\B & = t \T_2 + (1-t) \mu_2 \circ \D \, , \label{eq:B-spes}
\end{align}
where $\T_1$ and $\T_2$ are arbitrary trivial observables.
Since we have added an extra limitation to the conditions \eqref{eq:A-gen}--\eqref{eq:B-gen}, we cannot be sure anymore that a pair of compatible observables have this kind of representation.
However, if $w(\A;\trivial) \geq 1-t$ and $w(\B;\trivial) \geq t$, then by the definition of noise content we can find suitable observables $\C$ and $\D$ such that \eqref{eq:A-spes}--\eqref{eq:B-spes} hold.

As a conclusion, we obtain the following result and its equivalent formulation.

\begin{proposition}\label{prop:compatibility}
If $\A$ and $\B$ are two observables such that $w(\A;\trivial)+w(\B;\trivial) \geq 1$, then they are compatible.
\end{proposition}

\begin{proposition}\label{prop:incompatibility}
If $\A$ and $\B$ are incompatible observables, then $w(\A;\trivial)+w(\B;\trivial) < 1$.
\end{proposition}

The joint measurement scheme has a direct generalization for any finite number of observables. Let us consider $m$ observables $\A^{(1)},\ldots,\A^{(m-1)}$ and $\A^{(m)}$. We can then generalize conditions \eqref{eq:A-spes}--\eqref{eq:B-spes} to
\begin{align}
\A^{(j)} & = p_j \nu_j \circ \C^{(j)} + (1-p_j) \T^{(j)}, \label{eq:m-spes}
\end{align}
where $\T^{(j)}$ is an arbitraty trivial observable for each $j=1, \ldots, m$ and $p_j$ is an arbitrary probability distribution. As above, if $w(\A^{(j)};\trivial)\geq 1-p_j$ for all $j$ we can make \eqref{eq:m-spes} hold. By summing over $j$ we conclude the following generalization of Prop. \ref{prop:incompatibility}.

\begin{proposition}\label{prop:incompatibility-m}
If $\A^{(1)},\A^{(2)},\ldots,\A^{(m)}$ are incompatible observables, then $w(\A^{(1)};\trivial)+\ldots +w(\A^{(m)};\trivial)< m-1$.
\end{proposition}

%%%%%%%%%%%%%
\section{Applications of the incompatibility condition}\label{sec:specific}
%%%%%%%%%%%%%

%%%%%%%%%%%%%
\subsection{Eigenvalue condition for POVMs}
%%%%%%%%%%%%%

If $\A$ is an observable in finite dimensional quantum theory described by a POVM, we have that
\begin{equation}
\inf_{s \in \state} \A_x(s)= \min_{\psi \neq 0} \dfrac{\ip{\psi}{\A_x \psi}}{\ip{\psi}{\psi}} \, .
\end{equation}
It follows that $\inf_{s \in \state} \A_x(s)$ is the smallest eigenvalue of the effect operator $\A_x$.
Hence, by Prop. \ref{prop:inf} we conclude that $w(\A;\trivial)$ is the sum of the minimal eigenvalues of operators $\A_x$.
Combining this with Prop. \ref{prop:incompatibility-m}, we reach the following necessary condition for incompatibility.

\begin{corollary}\label{prop:eigen}
If $\A^{(1)},\ldots,\A^{(m)}$ is a collection of $m$ incompatible POVMs, then
the sum of the minimal eigenvalues of all their effects is smaller than $m-1$.
\end{corollary}

We will next illustrate the use of Cor. \ref{prop:eigen} in the
case of reverse observables. Consider a \emph{regular rank-1
POVM} $\A$, i.e., the effects of $\A$ read $\A_x = \frac{d}{N}
P_x$, where $d$ is the dimension of the Hilbert space, $N$ is the
number of outcomes and $P_x$ is a one-dimensional projection.
Examples of regular rank-1 POVMs include all nondegenerate sharp
POVMs and symmetric informationally complete POVMs.

As before, we denote by $\A^{\rm r} = \nu^{\rm r} \circ \A$ the
reverse version of $\A$. If $\A$ is a regular rank-1 POVM, then
the smallest eigenvalue of each operator $\A^{\rm r}_x$ is
$\frac{N-d}{N(N-1)}$. Applying Cor. \ref{prop:eigen}, we conclude
that the reverse versions of $m$ regular rank-1 POVMs with $N$
outcomes are compatible if
\begin{equation}\label{eq:reversed}
N \geq (d-1)\cdot m + 1 \, .
\end{equation}
It follows from this observation that, for instance, the reverse
versions of two regular rank-1 POVMs in $d=2$ are compatible for
all $N\geq 3$. One can readily find POVMs with two outcomes whose
reverse versions are incompatible; this is the case whenever the
original ones are incompatible since, in the case of two outcomes,
reversing is a reversible classical channel. Since the reversing
channel is more and more noisy when the number of outcomes
increases, one may wonder if there are any incompatible
collections of reverse POVMs when the number of outcomes is more
than two. In the following example we present a triplet of regular
rank-1 POVMs whose reverse versions are incompatible; the simple
compatibility condition \eqref{eq:reversed} is hence not trivial.

\begin{example}[\emph{Incompatible reverse POVMs}]
Consider three orthonormal bases $\{ \varphi_i \}_{i=1}^3$, $\{
\psi_i \}_{i=1}^3$, and $\{ \chi_i \}_{i=1}^3$ in a
three-dimensional Hilbert space $\hi_3$ such that a set
$\{\varphi_i,\psi_j,\chi_k \}$ is linearly independent for all
fixed $i,j,k$. Let $\A$, $\B$ and $\C$ be the POVMs related to
these bases, i.e., $\A_i = \ket{\varphi_i} \bra{\varphi_i}$, $\B_i
= \ket{\psi_i} \bra{\psi_i}$ and $\C_i = \ket{\chi_i}
\bra{\chi_i}$. The fact that the reverse POVMs $\A^{\rm r},\B^{\rm
r},\C^{\rm r}$ are incompatible can be proven by a contradiction.
Suppose $\A^{\rm r},\B^{\rm r},\C^{\rm r}$ are compatible, so that
there exists a joint POVM $\G$ with elements $\G_{ijk}$ such that
$\A^{\rm r}_i = \sum_{jk} \G_{ijk}$, $\B^{\rm r}_j = \sum_{ik}
\G_{ijk}$, and $\C^{\rm r}_k = \sum_{ij}\G_{ijk}$. As
$\ip{\varphi_i}{\A^{\rm r}_i \varphi_i} = 0$ and all the operators
$\G_{ijk}$ are positive, we have $\ip{\varphi_i}{\G_{ijk}
\varphi_i} = 0$ and this further implies $\G_{ijk} \varphi_i = 0$.
Similarly, $\G_{ijk} \psi_j = 0$ and $\G_{ijk} \chi_k= 0$. Since
the set $\{\varphi_i,\psi_j,\chi_k \}$ spans $\hi_3$, we conclude
that $\G_{ijk} = 0$. This contradicts the normalization
$\sum_{ijk} \G_{ijk} = \id$. Hence, the three POVMs $\A^{\rm
r},\B^{\rm r},\C^{\rm r}$ are incompatible.
\end{example}

The sufficient condition \eqref{eq:reversed} for compatibility of
the reverse versions of regular rank-1 POVMs is not necessary. We
will next demonstrate that there are compatible observables that
do not satisfy \eqref{eq:reversed}.

\begin{example}[\emph{Two mutually unbiased bases}]
Consider a $d$-dimensional Hilbert space $\hi_d$ and an
orthonormal basis $\{\varphi_i\}_{i=0}^{d-1}$ in it.
We denote $\omega = e^{i 2 \pi / d}$ and define another orthonormal basis
$\{\psi_j\}_{j=0}^{d-1}$ by
\begin{equation}
\psi_j = \frac{1}{\sqrt{d}} \sum_{k=0}^{d-1} \omega^{jk} \varphi_k
\, .
\end{equation}
These two bases are mutually unbiased, meaning that $|\ip{\varphi_i}{\psi_j}| =
\frac{1}{\sqrt{d}}$ for all $i,j=0,\ldots,d-1$.
The related POVMs $\A_i = \kb{\varphi_i}{\varphi_i}$ and
$\B_j=\kb{\psi_j}{\psi_j}$ consists of noncommuting projections and are hence incompatible.

The reverse versions  $\A^{\rm r}$ and $\B^{\rm r}$ are
incompatible if $d=2$, since then $\A^{\rm r}$ and $\B^{\rm r}$
are just relabelings of $\A$ and $\B$. However, for any $d\geq 3$,
$\A^{\rm r}$ and $\B^{\rm r}$ are compatible even if the
inequality \eqref{eq:reversed} does not hold. To see this, we
recall that by Prop. 2 in \cite{carmeli-2012}, $\A^{\rm r}$ and
$\B^{\rm r}$ are compatible whenever there exists a quantum state
$\sigma \in \sh$ such that
\begin{align}
{\rm tr}[\A_i \sigma] =
\frac{1-\delta_{i0}}{d-1} \quad \textrm{and} \quad {\rm tr}[\B_j \sigma] =
\frac{1-\delta_{j0}}{d-1} \, .
\end{align}
It is not hard to check that the
operator
\begin{eqnarray*}
\sigma &=& \frac{1}{d-1} \sum_{i=1}^{d-1}
\ket{\varphi_i}\bra{\varphi_i}
\nonumber\\
&& - \frac{1}{(d-1)(d-2)} \sum_{1 \leq i < j \leq d}
(\ket{\varphi_i}\bra{\varphi_j} + \ket{\varphi_j}\bra{\varphi_i})
\end{eqnarray*}
is a density operator and satisfies the conditions above.
Therefore, $\A^{\rm r}$ and $\B^{\rm r}$ are compatible.
\end{example}

As explained in Example \ref{ex:doubly}, the reversing channel $\nu^{\rm r}$ can also be applied to an
already reverse observable $\A^{\rm r}$ to obtain a doubly reverse observable $\A^{\rm rr}$.
It is not hard to see
from Prop. \ref{prop:eigen} that two doubly reverse observables
are always compatible if their number of outcomes $N \geq 3$. More
generally, a sufficient condition for compatibility of $m$ doubly
reverse observables with $N$ outcomes each is $m \leq (N-1)^2$.

%%%%%%%%%%%%%
\subsection{Eigenvalue condition for PPOVMs}
%%%%%%%%%%%%%

Let $\A$ be a PPOVM with an outcome set $X$ and the normalization
$\sum_{x \in X} \A_x = \varrho \otimes \id$ for some state
$\varrho$. We denote by $m_x$ the minimal eigenvalue of the PPOVM
element $\A_x$ for each $x\in X$. The noise content of $\A$
satisfies
\begin{equation}\label{PPOVM-noise-content}
w(\A;\trivial) \geq \sum_{x \in X} m_x \, .
\end{equation}
To see this, we define a trivial PPOVM $\T$ as
\begin{equation}
\label{eq:ppovm-construction}
\T_x = \frac{m_x}{m} \varrho \otimes \id \, ,
\end{equation}
where $m=\sum_{x \in X} m_x$.
Since
\begin{equation}
\A_x \geq m_x \id \otimes \id \geq m_x \varrho \otimes \id \, ,
\end{equation}
we can define
\begin{equation}
\A'_x = \frac{1}{1-m} (\A_x - m_x \varrho \otimes \id)
\end{equation}
and $\A'$ is a valid PPOVM.
We can then write
\begin{equation}
\A = m \T + (1-m) \A' \, ,
\end{equation}
which confirms \eqref{PPOVM-noise-content}.
Prop. \ref{prop:incompatibility-m} thus implies the following result, analogous to Cor. \ref{prop:eigen}.

\begin{corollary}\label{prop:eigen-ppovm}
If $\A^{(1)},\ldots,\A^{(m)}$ is a collection of $m$ incompatible PPOVMs, then the sum of the minimal eigenvalues of all their effects is smaller than $m-1$.
\end{corollary}

We note that in contrast to the case of POVMs, the eigenvalue formula~\eqref{PPOVM-noise-content} provides only a lower bound for the noise content of a PPOVM.
For instance, let
\begin{equation}\label{eq:only}
\A_x = p_x \ket{\psi_x}\bra{\psi_x} \otimes \id \, ,
\end{equation}
where $\ip{\psi_x}{\psi_{y}} = \delta_{xy}$ and $p_x$ is a probability distribution.
Then $m_x = 0$ for all $x$ and
the right hand side of \eqref{PPOVM-noise-content} equals 0.
But the PPOVM $\A$ is trivial, so that the left hand side of \eqref{PPOVM-noise-content} equals 1.

%%%%%%%%%%%%%
\subsection{Polytope state spaces}
%%%%%%%%%%%%%

A compact convex subspace $P$ of a finite dimensional vector space
$V$ is a \emph{polytope} if it has a finite number of extreme
elements. Let $\mathrm{ext}(P)=\{s_1, \ldots, s_n\}$ be the set of
extreme elements of a polytope $P$. Since every state $s \in P$
can be represented as a convex sum of elements in
$\mathrm{ext}(P)$, we have that
\begin{align*}
\A_x(s)&=\A_x\left(\sum_{i} \lambda_i s_i\right) = \sum_i \lambda_i \A_x(s_i) \\
&\geq \sum_i \lambda_i \min_{k} \A_x(s_k) = \min_{k} \A_x(s_k)
\end{align*}
for every $s \in P$, and thus $\inf_{s \in \state} \A_x(s) = \min_{s \in \mathrm{ext}(P)} \A_x(s)$.
Combining this with Prop. \ref{prop:incompatibility-m}, we get an analogous result to the previous eigenvalue conditions for POVMs and PPOVMs.
\begin{corollary}\label{cor:poly}
If $\A^{(1)},\ldots,\A^{(m)}$ is a collection of $m$ incompatible observables on a polytopic state space $P$, then
the sum of minimal values of all of their effects on $\mathrm{ext}(P)$ is smaller than $m-1$.
\end{corollary}

%%%%%%%%%%%%%%%%%%%%%%%%%%%%%%%%%%%%%%%%%%%%%%%%%%%%%%%%%%%%%%%%%%%%
\begin{figure}
\centering
\includegraphics[scale=1.0]{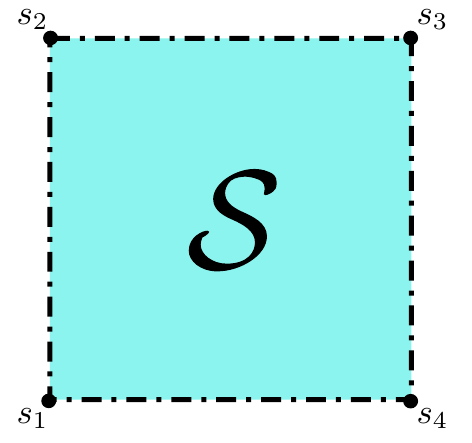}
\caption{\label{figure5} Squit state space.}
\end{figure}
%%%%%%%%%%%%%%%%%%%%%%%%%%%%%%%%%%%%%%%%%%%%%%%%%%%%%%%%%%%%%%%%%%%%

In the following, we take $\state$ to be a state space that is
isomorphic to a square in $\real^2$, i.e., to the convex hull of
four points $s_1,s_2,s_3,s_4 \in \real^2$ satisfying
$s_1+s_3=s_2+s_4$ (see Fig. \ref{figure5}). This is called the
square bit state space, or squit state space for short.

We consider a class of binary observables $\A^{\alpha}$ and
$\B^{\beta}$, parametrized by $\alpha, \beta \in [0,1]$, whose
outcomes are labeled by $\pm$ and defined on the extreme points
$s_1$, $s_2$, $s_3$, and $s_4$ as
\begin{align*}
\A^{\alpha}_+(s_1) = \A^{\alpha}_+(s_2)=\alpha , \quad \A^{\alpha}_+(s_3)=\A^{\alpha}_+(s_4)=1, \\
\B^{\beta}_+(s_1) = \B^{\beta}_+(s_4)=\beta, \quad
\B^{\beta}_+(s_2)=\B^{\beta}_+(s_3)=1.
\end{align*}
The values of $\A^\alpha$ and $\B^\beta$ are depicted in Fig.
\ref{figure6}.

%%%%%%%%%%%%%%%%%%%%%%%%%%%%%%%%%%%%%%%%%%%%%%%%%%%%%%%%%%%%%%%%%%%%
\begin{figure}
\centering
\includegraphics[scale=1.0]{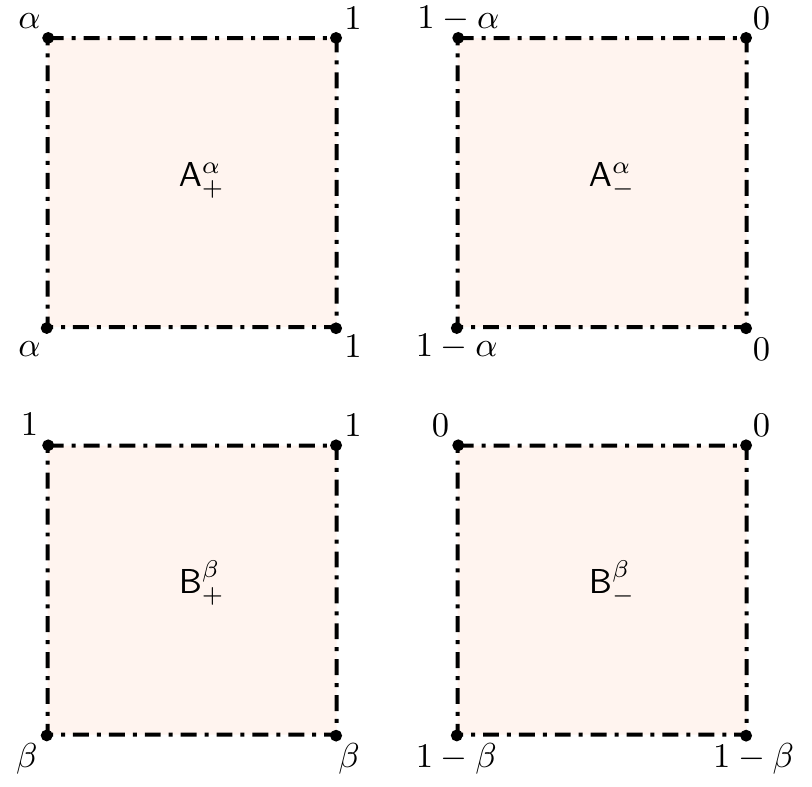}
\caption{\label{figure6} Observables $\A^\alpha$ and $\B^\beta$.}
\end{figure}
%%%%%%%%%%%%%%%%%%%%%%%%%%%%%%%%%%%%%%%%%%%%%%%%%%%%%%%%%%%%%%%%%%%%

Now we see that
\begin{align*}
w(\A^{\alpha};\trivial) = \min_{s \in \mathrm{ext}(\state)} \A_+(s)+\min_{s \in \mathrm{ext}(\state)} \A_-(s) =  \alpha \, ,
\end{align*}
and similarly that $w(\B^{\beta};\trivial)= \beta$. Hence, by Cor. \ref{cor:poly}, if
\begin{equation}
\alpha + \beta \geq 1,
\label{eq:alpha-beta}
\end{equation}
then observables $\A^{\alpha}$ and $\B^{\beta}$ are compatible.
It is easy to find $\A^{\alpha}$ and $\B^{\beta}$ as mixtures with maximal noise contents,
\begin{align*}
\A^{\alpha} &= \alpha \T + (1-\alpha) \A \\
\B^{\beta} &= \beta \T + (1-\beta) \B, \\
\end{align*}
where $\T$ is the trivial binary observable with $\T_+(s)=1$ and $\T_-(s)=0$ for all $s \in \state$, and $\A \equiv \A^0$ and $\B \equiv \B^0$.

The observables $\A$ and $\B$ are themselves incompatible.
Moreso, they are maximally incompatible in the sense that the minimum amount of noise one has to mix them with to make their noisy versions compatible is enough to make any other pair of observables compatible.
More precisely, it was shown in \cite{busch-2013} that the observables $\lambda \A + (1-\lambda)\T_1$ and $\mu \B + (1-\mu) \T_2$ are incompatible for all choices of trivial observables $\T_1$ and $\T_2$ if and only if $\lambda + \mu > 1$.
Therefore, we conclude that the inequality \eqref{eq:alpha-beta} derived from Prop. \ref{prop:compatibility} is actually both necessary and sufficient for the compatibility of $\A^\alpha$ and $\B^\beta$.

\ \\
%%%%%%%%%%%%%%%%%%%%%%%%%%%%%%%%%
\section{Conclusions}

We have considered general probabilistic theories on an equal
footing and quantified the noise content of observables in every such
theory via the set of trivial observables. In the case of standard quantum
theory, the noise content is merely the sum of minimal eigenvalues
of the POVM effects. In the quantum theory of processes, the noise
content is bounded below by the sum of minimal eigenvalues of the
corresponding PPOVM effects. In general, the noise content can be
quantified with respect to any subset of observables.

We have derived the noise content inequality for a pair of
observables, which is a necessary condition for their
incompatibility. Our approach is based on a modification of the
adaptive strategy for building a joint observable. We have then
extended this result to the case of $m$ observables. By way of examples
with reverse regular observables we have demonstrated
non-triviality of the derived noise content inequality. Moreover,
this inequality turned out to not only be necessary but also
sufficient for incompatibility of some observables in the square bit
state space.

%%%%%%%%%%%%%%%%%%%%%%%%%%%%%%%%%
\section*{Acknowledgements}

The authors wish to thank Michal Sedl\'{a}k and M\'{a}rio Ziman
for clarifying discussions concerning PPOVMs and Tom Bullock on useful comments on the manuscript. S.N.F. acknowledges the support of Academy of Finland for a mobility grant to conduct
research in the University of Turku, where this article was
initiated. S.N.F. is grateful to the University of Turku for kind
hospitality. S.N.F. thanks the Russian Foundation for Basic
Research for partial support under Project No. 16-37-60070
mol-a-dk.

%%%%%%%%%%%%%%%%%%%%%%%%%%
%%%%%%%%%%%%%%%%%%%%%%%%%%
%%%%%%%%%%%%%%%%%%%%%%%%%%


\begin{thebibliography}{99}
\bibitem{plavala-2016}
M. Pl\'{a}vala, All measurements in a probabilistic theory are
compatible if and only if the state space is a simplex,
Phys. Rev. A {\bf 94}, 042108 (2016).
\end{thebibliography}
\end{document}